\documentclass[conference,a4paper]{IEEEtran}
\IEEEoverridecommandlockouts
\usepackage{cite}
\usepackage{amsmath,amssymb,amsfonts}
\usepackage{algorithmic}
\usepackage{graphicx}
\usepackage{textcomp}
\usepackage{xcolor}
\def\BibTeX{{\rm B\kern-.05em{\sc i\kern-.025em b}\kern-.08em
    T\kern-.1667em\lower.7ex\hbox{E}\kern-.125emX}}
    
\usepackage{enumitem}  
\newtheorem{theorem}{Theorem}

\newtheorem{remark}{Remark}

\newtheorem{lem}{Lemma}
\newtheorem{definition}{Definition}
\newtheorem{assum}{Assumption}
\newenvironment{proof}{\textit{Proof: }}{\hfill$\blacksquare$}

\begin{document}

\title{A fair Peer-to-Peer Electricity Market model for Residential Prosumers\\
\thanks{This work was partially supported by NWO under research project P2P-TALES (grant n. 647.003.003) and the ERC under research project COSMOS, (802348).}
}

\author{\IEEEauthorblockN{ Aitazaz Ali Raja}
\IEEEauthorblockA{\textit{Delft Center for Systems and Control,}
\textit{TU Delft}\\
The Netherlands \\
\texttt{a.a.raja@tudelft.nl}}
\and
\IEEEauthorblockN{ Sergio Grammatico}
\IEEEauthorblockA{\textit{Delft Center for Systems and Control,} 
\textit{TU Delft}\\
The Netherlands \\
\texttt{s.grammatico@tudelft.nl}}
}

\IEEEoverridecommandlockouts
\IEEEpubid{\makebox[\columnwidth]{978-1-6654-4875-8/21/\$31.00~\copyright 2021 IEEE\hfill} \hspace{\columnsep}\makebox[\columnwidth]{ }}

\maketitle
\IEEEpubidadjcol

\begin{abstract}
In this paper, we propose a bilateral peer-to-peer (P2P) energy trading scheme for residential prosumers with a simplified entry to the market. We formulate the market as an assignment game, a special class of coalitional games. For solving the resulting decision problem, we design a bilateral negotiation mechanism that enables matched buyer-seller pairs to reach a consensus on a set of ``stable" and ``fair" trading contracts. The proposed negotiation process can be executed on possibly time-varying communication networks with virtually minimal information requirements that in turn preserves privacy among prosumers. Numerical simulations illustrate the beneficial features of our P2P market model and negotiation protocol. 
\end{abstract}


\section{Introduction}
Decarbonization of energy systems is one of the key agenda of the climate action plan and to achieve this goal, power systems are envisioning large scale integration of renewable energy sources (RES). Wide scale decentralized deployment of RES, especially photo-voltaic (PV), is being undertaken by the small prosumers (e.g. residential) which brings them at the center of this transformation \cite{cai2013impact}. Thus, many demand-side tools are being developed for the technical and economic integration of the residential prosumers.

Local or community based electricity markets can effectively facilitate the distributed deployment of RES by managing the associate uncertainty locally and by providing financial benefits. Therefore, such market based solutions have received considerable attention from smart-grid researchers, especially towards the peer-to-peer (P2P) market paradigm \cite{Tushar2020}.  P2P  markets  provide prosumers with  the direct control over the trade of their energy sources on their own terms of transactions to make profitable interactions. Thus,  it  encourages  wider  prosumer participation and also provides significant  benefits  to  the  system  operators  for  example  in terms of peak shaving \cite{moret2018energy}, and lower investments in grid capacity \cite{mengelkamp2018designing}.

However, the design of  the local P2P electricity  markets presents  mathematical and structural challenges. The whole sale market in EU requires sellers to enter with the complex offers, which requires high level of technical abilities. Such structure cannot be replicated in the markets where the participants are laypersons (residential prosumers). Thus, the key tasks are to design: a mechanism which seeks a market equilibrium while incorporating a self-interested  decision-making  attitude  by  the  participants and a structure simple enough to encourage the entry by residential prosumers. Therefore, in this paper, we first design a bilateral P2P market that simplifies the entry of a typical residential prosumer and then we present an algorithm that enables prosumers to converge to a \textit{fair} and \textit{stable} market solution, in context of coalitional game theory.

Coalitional game theory provides analytical tools to  study the cooperative interaction of selfish  and  rational  agents and thus, holds adequate prospects for the design of P2P markets. The authors in \cite{tushar2018peer} propose a P2P energy trading scheme in which prosumers form a coalition to trade energy among themselves, at a mid-market rate which ensures the stability of the coalition.  In \cite{han2018incentivizing} the authors use coalitional game theory to formulate a community based architecture for local energy exchange to minimize overall energy cost. A coalition formation game is formulated by the authors in \cite{tushar2019grid} for P2P energy exchange among prosumers, and the resulting coalition structure is shown to be stable. The price of exchange is determined by the double auction mechanism. Another coalition formation game is presented by the authors in \cite{tushar2020coalition}, that allows prosumers to optimize their battery usage for P2P energy trading. In \cite{khorasany2020new} the authors present an iterative procedure for peer matching among prosumers, who then undertake a bilateral negotiation to come to an agreement on the price and quantity of energy trade, but without coalitional game theoretic guarantees.
In this paper, we model P2P energy trading as an assignment game, a special class of coalitional games, that allows for bilateral contracts and for a mutual settlement on the \textit{fair} and \textit{stable} contract prices via distributed negotiation.

\textit{Contribution}: We propose a bilateral P2P electricity market and a solution mechanism within the framework of coalitional game theory. Our key contributions are summarized next:
\vspace{-1mm}
\begin{itemize}
    \item We formulate bilateral P2P energy trading as an assignment game which simplifies the prosumer participation, amidst of RES uncertainty, and allows for product differentiation. Our formulation ensures the existence of a ``stable" set of bilateral contracts (Section \ref{sec: Peer-to-Peer market as coalitional game});
    \item We develop a distributed negotiation mechanism where the buyer  (seller) communicates only with the matched seller (buyer) and the market operator, and  we show that the mechanism converges to the $\tau$-value in the core of the assignment game while preserving the privacy among the market participants (Section \ref{sec: solution mechanism});
\end{itemize}

 \textit{Notation and definitions}: Given a mapping $M: \mathbb{R}^n \rightarrow \mathbb{R}^n, \mathrm{fix}(M):= \{x \in \mathbb{R}^n \mid x = M(x)\} $ denotes the set of its fixed points. $\text{Id}$ denotes the identity operator. For a closed set \(C \subseteq \mathbb{R}^{n},\) the mapping $\mathrm{proj}_C$: \(\mathbb{R}^{n} \rightarrow C\) denotes the projection onto \(C,\) i.e., \(\operatorname{proj}_C(x)=\) \(\arg \min _{y \in C}\|y-x\| .\)  \(A \otimes B\) denotes the Kronecker product between the matrices \(A\) and \(B .\) {For $x_{1}, \ldots, x_{N} \in \mathbb{R}^{n},$ $\mathrm{col}(\left(x_{i}\right)_{i \in(1, \ldots, N)}):=\left[x_{1}^{\top}, \ldots, x_{N}^{\top}\right]^{\top}.$} $\mathrm{dist}(x,y):= \|y-x\|$. {For a closed set \(C \subseteq \mathbb{R}^{n}\) and $N \in \mathbb{N}, C^N:=\prod_{i=1}^{N} C_{i}$}. A continuous mapping $M : \mathbb{R}^n \rightarrow \mathbb{R}^n$ is a paracontraction, with respect to a norm $\|\cdot\|$ on $\mathbb{R}^n$, if $ \|M(x) - y\| < \|x - y\|,$ for all $(x,y) \in (\mathbb{R}^n \backslash \mathrm{fix}(M)) \times \mathrm{fix}(M)$.
\section{Mathematical background on assignment games and the $\tau$-value}
\subsection{Assignment games}
Assignment games are a class of coalitional games  that model a two sided matching market with the objective of finding optimal assignments between the opposite sides, for example buyers and sellers \cite{shapley1971assignment}. Let us denote the two sets of agents, i.e., buyers and sellers by $\mathcal{I}_\mathcal{B}$ and $\mathcal{I}_\mathcal{S}$, respectively. Here, each seller $j \in \mathcal{I}_\mathcal{S}$ offers a good at the price of at least $c_j$ and each buyer $i \in \mathcal{I}_\mathcal{B}$, bids $h_{i,j}$ for the good of seller $j$.  Then, the value function that gives value to each buyer-seller pair, with a slight abuse of notation, reads as:
\begin{equation}\label{eq: value function}
v(i,j) = \max\{0,h_{i,j} - c_j\}. 
\end{equation}
Any viable assignment  must satisfy $h_{i,j} > c_i$. Furthermore,  one-sided coalitions generate no value, i.e., $v(S) = 0$ if $S \subseteq \mathcal{I}_\mathcal{B}$ or $S \subseteq \mathcal{I}_\mathcal{S}$, thus only mixed coalitions are meaningful.

 Utilizing the fact that the buyer-seller pairs alone suffice to determine the two-sided market completely, we define an assignment matrix $M = [v(i,j)]$ for all pairs $(i,j) \in \mathcal{I}_\mathcal{B}\times \mathcal{I}_\mathcal{S}$.

\begin{definition} [Value function]\label{def: value function}
Let $\mathcal{I}_\mathcal{B} = \{1, \ldots, N_\mathcal{B}\}$ and $\mathcal{I}_\mathcal{S}= \{1, \ldots, N_\mathcal{S}\}$ be the sets of buyers and sellers, respectively. Let $M = [v(i,j)]_{(i,j) \in \mathcal{I}_\mathcal{B}\times \mathcal{I}_\mathcal{S}}$ be an assignment matrix with $v(i,j)$ as in (\ref{eq: value function}). Given $\mathcal{B} \subseteq \mathcal{I}_\mathcal{B}$ and $\mathcal{S} \subseteq \mathcal{I}_\mathcal{S}$, let $\mathcal{P}(\mathcal{B}, \mathcal{S})$ be the set of all possible matching configurations between $\mathcal{B}$ and $\mathcal{S}$, where a matching configuration is a set of two-sided matchings such that a seller (buyer) is matched with at most one buyer (seller). Then, the value function $v_{M}: \mathcal{I}_\mathcal{B} \cup \mathcal{I}_\mathcal{S} \to \mathbb{R}$ is defined as, $v_{M}(\mathcal{B} \cup \mathcal{S})=\underset{P \in \mathcal{P}(\mathcal{B}, \mathcal{S})}{\max} \sum_{(i,j) \in P} v(i,j).$ $\hfill \square$
\end{definition} 
Let us now formally define an assignment game.
 \begin{definition} [Assignment game]\label{def: assignment game}
Let $\mathcal{I}_\mathcal{B} = \{1, \ldots, N_\mathcal{B}\}$ and $\mathcal{I}_\mathcal{S}= \{1, \ldots, N_\mathcal{S}\}$ be the sets of buyers and sellers, respectively.  An assignment  game is a pair $\mathcal{M} = (\mathcal{I}_\mathcal{B} \cup \mathcal{I}_\mathcal{S}, v_M)$ such that the value function $v_M$ is as in Definition \ref{def: value function}. 
 $\hfill \square$
\end{definition} 
In an assignment game, the value generated by an optimal match between a buyer-seller pair $(i,j)$, $v(i,j)$, is distributed between both the members as a payoff.
\begin{definition}  [Payoff vector] \label{def: payoff}
Let $\mathcal{M} = (\mathcal{I}_\mathcal{B} \cup \mathcal{I}_\mathcal{S}, v_M)$ be an assignment game and $ (i,j) \in \mathcal{I}_\mathcal{B} \times \mathcal{I}_\mathcal{S}$ be a matched pair. Then, the element $(x'_i, x''_j)$, of vectors $(x', x'') \in \mathbb{R}^{N_\mathcal{B}} \times \mathbb{R}^{N_\mathcal{S}}$ represents the share of agent $i$ and $j$ of the value $v (i,j)$. $\hfill \square$
\end{definition}
The most widely studied solution concept of assignment games is the core, a set of \textit{efficient} and \textit{rational} payoff vectors. The core relates to the stability of  assignment, i.e., the satisfaction of the self-interested agents with the corresponding match.
\begin{definition} [Core of assignment game]
The core $\mathcal{C}_M$ of an assignment game $(\mathcal{I}_\mathcal{B} \cup \mathcal{I}_\mathcal{S}, v_M)$ is the following set:
\begin{align}\label{core assignment}
 & \mathcal{C}_M :=  \big\{ (x', x'') \in \mathbb{R}^{N_\mathcal{B}} \times \mathbb{R}^{N_\mathcal{S}} \mid \sum_{i \in \mathcal{I}_\mathcal{B}} x_i' + \sum_{j \in \mathcal{I}_\mathcal{S}} x''_j = \nonumber\\
  & \qquad v_M(\mathcal{I}_\mathcal{B} \cup \mathcal{I}_\mathcal{S}), x_i' + x''_j \geq \hat{v}(i,j) \text{ for all } (i,j) \in \mathcal{I}_\mathcal{B} \times \mathcal{I}_\mathcal{S}  \big\}. 
\end{align}

$\hfill \square$
\end{definition}

    \begin{remark}[Non-emptiness of core {\cite{shapley1971assignment}}]\label{lem: non empty core}
An assignment game (as in Definition \ref{def: assignment game}) has a non-empty core. $\hfill \square$
\end{remark}

 The payoff $(x'_i, x''_j)$ as a result of a bilateral trade between an optimally matched pair $(i,j) \in \mathcal{I}_\mathcal{B} \times \mathcal{I}_\mathcal{S}$ determines the bilateral contract price $\lambda_{i,j}$. The contract price is defied as the difference of the bid of buyer $i$ and his payoff, i.e.,  $\lambda_{i,j} = h_{i,j} - x'_i$.
We remark that, the  core  set  is  not  singleton and different core payoffs can favor different sides of the market \cite{shapley1971assignment}. The two extreme points of the core are referred as the buyer optimal payoff $(\overline{x}', \underline{x}'')$ and the seller optimal payoff $(\underline{x}', \overline{x}'')$. Therefore, it is important to identify a \textit{fair} payoff that belongs to the core. Next, we provide one such \textit{fair} payoff, namely the $\tau$-value in context of assignment games \cite{Nunez2002}.
\subsection{The $\tau$-value for assignment games}
The $\tau$-value is generally defined as an average of the utopic payoff and the minimal rights payoff where, the utopia payoff is regarded as the maximal payoff an agent can receive in the core and the minimal rights payoff is what an agent can guarantee himself. In context of an assignment game $(\mathcal{I}_\mathcal{B} \cup \mathcal{I}_\mathcal{S}, v_M)$, the utopia payoff for a buyer $i$ is given by his marginal contribution to the grand coalition, which is also the buyer optimal, i.e., $ \overline{x}'_i = v_M (\mathcal{I}_\mathcal{B} \cup \mathcal{I}_\mathcal{S}) - v_M (\mathcal{I}_\mathcal{B} \cup \mathcal{I}_{\mathcal{S}\backslash \{i\}})$.  Furthermore, the minimal rights payoff of buyer $i$ that is optimally matched with seller $j$ is given as $\underline{x}'_i = v_M (\mathcal{I}_\mathcal{B} \cup \mathcal{I}_{\mathcal{S}\backslash \{j\}}) - v_M (\mathcal{I}_{\mathcal{B}\backslash \{i\}} \cup \mathcal{I}_{\mathcal{S}\backslash \{j\}})$. Finally, the $\tau$-value of an assignment game is as follows:
\begin{equation}\label{eq: tau value}
 \tau (v_M) = \frac{(\overline{x}', \underline{x}'') + (\underline{x}', \overline{x}'')}{2}   
\end{equation}
\begin{remark}[$\tau$-value in core {\cite{Nunez2002}}]\label{rem: tau value in core}
The $\tau$-value of an assignment game (as in Definition \ref{def: assignment game}) belongs to the core. $\hfill \square$
\end{remark}
Using the fact that the $\tau$-value is a midpoint between the buyer-optimal and the seller-optimal core payoff, we can define a set of favorable payoffs for each side. 
\begin{definition} [Favorable payoff]
For buyer $i$ in an optimally matched pair $(i,j)$, the set of favorable payoffs is
\begin{equation} \label{eq: fav payoff}
 \mathcal{X}_i := \{x'_i \in \mathbb{R} \mid x'_i \geq \frac{(\overline{x}_i' + \underline{x}_i')}{2}, x'_i + x''_j = v(i,j) \}. 
\end{equation} 
$\hfill \square$
\end{definition}

In the sequel, we associate the idea of fairness to the $\tau$-value and use it as a solution concept for our bilateral P2P market design. In practice, bilateral agreements should be directly negotiated by self-interested agents. Thus, we propose a distributed solution mechanism in which the agents negotiate bilaterally to autonomously reach a consensus on the payoff. 
\begin{definition} [Consensus set]\label{def: consensus set}
The consensus set $\mathcal{A} \subset \mathbb{R}^{N^{2}}$ is defined as:
\begin{equation}\label{eq: consensus}
\mathcal{A} := \{\mathrm{col}(\boldsymbol{x}_1, \ldots, \boldsymbol{x}_N) \in \mathbb{R}^{N^{2}} \mid \boldsymbol{x}_i = \boldsymbol{x}_j, \forall i,j \in \mathcal{I}\}. 
\end{equation} $\hfill \square$
\end{definition}

\vspace{-3mm}
\section{P2P market as an assignment game}\label{sec: Peer-to-Peer market as coalitional game}

 In this section, we formulate a bilateral P2P energy market as assignment game where the participants are partitioned into buyers and sellers.  A prosumer who owns an energy source is regarded as a seller while a buyer can be a mere consumer as well. For modelling a seller, we consider a typical residential prosumer who is not present at home during high PV generation hours on the weekdays. Therefore, it makes high economic sense for such prosumer to sell the energy produced in the market. Let us elaborate on the models of the market participants namely, sellers, buyers and the market operator. First, a sellers' offer is composed of a rated power of the generation source and the price per KWh of energy for the period of availability. Such offer structure, with easily known parameters, greatly simplifies the process of prosumer's entry into the market which in fact is an important practical requirement for enabling the participation of residential prosumers (layman) in P2P markets. Another way of looking at our model is that a seller offers to rent out his generation source for the desired time period. Our market design also creates an opportunity for the data markets in energy systems that allows participants to share their generation and demand data for additional financial or operational benefits. This data is then utilized by the market operator to optimize amidst of uncertainty. We note that our model can accommodate other energy sources as well (e.g. ES) but we maintain our focus on PV as it is most widely adopted RES at residential level.
 
A buyer enters the market with the energy demand,  bid per KWh and the preference factor that allows a buyer to prioritise sellers based on the desired criteria (e.g. green energy, seller rating). Finally, the market is operated by a central  operator who has complete information of bids and offers, and is also responsible for maximizing the market welfare. 

Let $c_j$ denote the price demanded by a seller $j \in \mathcal{I}_\mathcal{S}$ per KWh of energy and let $s_j$ represent the rated power of the offered energy source; then, an offer of a seller $j$ is given by a pair $(c_j, s_j)$. Similarly, let us denote the energy demand of a buyer $i \in \mathcal{I}_\mathcal{B}$ by $d_i$ and his preference factor for seller $j$ by  $\alpha_{i,j}$. Furthermore, let $p_i$ be a base price that the buyer $i$ is willing to pay for each unit (1 KWh) of energy, hence he presents its bid as $(\alpha_{i,j}p_i, d_i)$. Next, we impose some practical limitations on the buyer bids and seller offers to make our market setup economically rational for the participants.
 Let $g_\text{b}$ and $g_\text{s}$ denote the buying price and the selling price of energy provided by the grid, respectively. Then, the rational buyer $i$ should offer a seller $j$ a higher energy price than that of the grid, but not more than the grid's selling price, i.e.,
$ \alpha_{i,j}p_i \in (g_\text{b}, g_\text{s}] $ and analogously, the seller $j$ satisfies $c_j \in [g_\text{b}, g_\text{s}). $ 

The energy generation by RES (PV) is inherently uncertain thus to encourage prosumer participation we transfer the responsibility of accounting for this uncertainty from a seller to the market operator. We achieve this by allowing seller to only include the rated power of his energy asset instead of the energy.  Therefore, a stochastic market mechanism is required. For this purpose we use scenario modeling of uncertainty in RES generation. Let the set of generation scenarios of future be $\mathcal{F}$ and denote the probability of occurrence of the scenario $f \in \mathcal{F}$ by $\rho_f$. Also, let us denote the generation forecast of seller $j$'s energy source  in scenario $f$ by $\hat{s}_j(f)$ then the corresponding expected value is given by $\mathbb{E}[\hat{s}_j] = \sum_{f \in \mathcal{F}} \rho_f \hat{s}_j(f) $. We note that, without the loss of generality, the uncertainty can also be considered in demand. However, in this paper we assume the demand forecast to be deterministic. Next, we formulate a P2P energy market as an assignment game.

In our P2P market setup, the sellers and buyers make bilateral contracts that generate certain utility (value) for both. Let buyer $i \in \mathcal{I}_\mathcal{B}$ and seller $j \in \mathcal{I}_\mathcal{S}$ make a bilateral contract then, the value $\hat{v}(i,j)$ generated by this contract reads as
\begin{equation}\label{eq: contract}
    \hat{v}(i,j) =
    \begin{cases}
    \begin{array}{ll}
    &(\max \{0, \alpha_{i,j}p_i - c_j\}) d_i \qquad \quad \; \text{if } \mathbb{E}[\hat{s}_j] \geq d_i \\

    &(\max \{0, \alpha_{i,j}p_i - c_j\}) \mathbb{E}[\hat{s}_j] \qquad  \text{otherwise.} \\
    \end{array} 
    \end{cases}
\end{equation}
Let us elaborate on the bilateral contract value in (\ref{eq: contract}). First, the contract is only viable when buyer's bid of the energy is higher than seller's offer, i.e., $\alpha_{i,j}p_i > c_j$. Then, in the first case, i.e., $\mathbb{E}[\hat{s}_j] \geq d_i$, trading each unit generates the welfare equal to $\alpha_{i,j}p_i - c_j$ where, the total traded units are $d_i$. Furthermore, the excess energy of the seller $(s_j - d_i)$ is sold to the grid. The second case has a similar explanation. We note that, the value of a non-viable contract will be zero and that if $s_j = d_i$ then the two cases are equivalent. 

Now, let us define an assignment matrix $M = [\hat{v}(i,j)]_{(i,j) \in \mathcal{I}_\mathcal{B}\times \mathcal{I}_\mathcal{S}}$ where each element $\hat{v}(i,j)$ represents the value of a bilateral contract between buyer $i \in \mathcal{I}_\mathcal{B}$ and seller $j \in \mathcal{I}_\mathcal{S}$. Then, the corresponding assignment game is given by $\mathcal{M}=\left(\mathcal{I}_{B} \cup \mathcal{I}_{S}, v_{M}\right)$. To solve the resulting game, the market operator first evaluates the value $v_M(S)$, for each $S \subseteq \mathcal{I}$, by solving the following assignment problem: 
\begin{equation}\label{eq: assignment game}
\mathbb{P}(S):
\left\{ \quad
\begin{aligned}
 &\displaystyle \underset{\mu}{\max}  \sum_{i \in \mathcal{I}_{\mathcal{B}}\cap S} \sum_{j \in \mathcal{I}_{\mathcal{S} }\cap S} \hat{v}(i,j) \mu_{i, j} \\
    & \displaystyle \:\: \mathrm{s.t.}  \sum_{i \in \mathcal{I}_{\mathcal{B}}\cap S}  \mu_{i, j} \leq 1 \qquad \quad \forall j \in \mathcal{I}_{\mathcal{S} }\cap S \\
    & \displaystyle \quad \;\;\:  \sum_{j \in \mathcal{I}_{\mathcal{S}}\cap S}  \mu_{i, j} \leq 1 \qquad \quad \forall i \in \mathcal{I}_{\mathcal{B} }\cap S 
\end{aligned}
\right.
\end{equation}
with matching factors $\mu_{i, j} \in \{0, 1\}$, where $\mu_{i, j} = 1$ represents the matching between buyer $i$ and seller $j$. The constraints imposed on the matching factors ensure one-to-one matching. We note that the sellers and buyers can also enter as multiple agents to ensure adequate energy trading in the case of participation discrepancy on two sides of the market. The results obtained by solving the assignment problem in (\ref{eq: assignment game}) enable the market operator to evaluate the optimal buyer-seller assignment and the marginal contribution of the agents. Following are the notable features of our P2P market design:  
\begin{itemize}
    \item Existence: There always exist a set of unobjectionable bilateral contracts for all participants, i.e., the core of a bilateral P2P market is always non-empty (Remark \ref{lem: non empty core}).
    \item Product differentiation: Buyers can assign priority to seller characteristics (e.g. green energy, location of seller) via preference factors $\alpha_{i,j}$.
    \item Convenience: Residential prosumers do not require any technical tools or methods to offer suitable amount of energy for given time interval (offer includes power rating) thereby simplifying the market entry.
   \item Bilateral negotiation: Optimal bilateral contracts are assigned centrally by the market operator but the contract price is negotiated internally between matched buyer and seller, thus preserving the inter-prosumer privacy.
\end{itemize}
\section{Solution mechanism}\label{sec: solution mechanism}
After the market operator's evaluation of the optimal assignment, members of each matched pair negotiate between themselves for a bilateral agreement on the trading price. The goal of our mechanism design is to enable the matched pairs to independently reach a consensus on a fair bilateral contract price such that the corresponding collective vector of payoffs belong to the core in (\ref{core assignment}). Thus, to achieve our goal, we propose a distributed bilateral negotiation mechanism, where a central market operator with complete information of the game transmits to the market participants only their marginal contribution to the grand coalition (P2P market). After receiving the required information, each agent distributedly proposes a payoff allocation to his matched agent. We prove that utilizing such a limited information, the proposed solution converges to a fair and stable payoff distribution, i.e, to the $\tau$-value.
\vspace{-1mm}
\subsection{distributed bilateral negotiation}\label{subsec: Solution mechanism}
We consider a bilateral negotiation process in which, at each negotiation step $k$, a buyer (seller) communicates with the matched seller (buyer) to bargain for his payoff.  We present the process of negotiation for a matched pair $\mathcal{E}_{i,j} \in \mathcal{E} := \{(i,j) \mid \mu_{i,j} = 1, \text{ for all } (i,j) \in \mathcal{I}_\mathcal{B}\times \mathcal{I}_\mathcal{S}\}$ where, $\mathcal{E}$ is a set of matchings in an optimal assignment attained by solving (\ref{eq: assignment game}) for the grand coalition $(\mathcal{I}_{B} \cup \mathcal{I}_{S})$.
The matched pair communicates over a directed network link 
weighted using an adjacency matrix $W^{k} = [w_{i,j}^k]$, whose element $w_{i,j}^k$ represents the weight assigned by agent $i$ to the payoff proposal of his matched agent $j$, ${ \boldsymbol{x}}_j^{k}$. Here, the time-variation $k$ refers to the variation of the weights assigned by each agent to the proposal of the corresponding matched agent.
Furthermore, we assume the adjacency matrix to be stochastic with the positive entries, which means that an agent always gives some weight to his previous proposal and the proposal of the matched agent.
\begin{assum}[Stochastic adjacency matrix]\label{asm: graph}
 For all $k \geq 0$, the adjacency matrix $W^{k} = [w_{i,j}^k]$ of the communication links is row-stochastic and $\exists$ $\gamma > 0$ such that $w_{i,j}^k \geq \gamma$. $\hfill \square$
\end{assum}
We further make a technical assumption on the elements of the adjacency matrix $W^{k}$, i.e., they belong to a finite set hence, finitely many adjacency matrices are available.
\begin{assum}[Finitely many adjacency matrices]\label{asm: fixed graph}
The adjacency matrices $\{W^k\}_{k \in \mathbb{N}}$ of the communication graphs belong to $\mathcal{W}$, a finite family of matrices that satisfy Assumption \ref{asm: graph}, i.e., $W^k \in \mathcal{W}$ for all $k \in \mathbb{N}$.  $\hfill \square$
\end{assum}

At each negotiation step $k$, an agent $i$ bargains by proposing a payoff distribution $\boldsymbol{x}_i^k \in \mathbb{R}^2$, for both of the agents in a matched pair $(i,j)$. To evaluate a proposal, he first takes an average of the estimate of the matched agent $\boldsymbol{x}_j^k$ and his own proposal weighted by an adjacency matrix $W^k$, $\sum_{j \in \mathcal{E}_{i,j}} w_{i,j}^k \boldsymbol{x}_{j}^{k}$.  Next, agent $i$ receives the information of his marginal contribution in the market by the market operator, which allows agent $i$ to evaluate the set of his favorable payoffs $\mathcal{X}_i$, as in (\ref{eq: fav payoff}). We note that our algorithm does not require evaluation of the complete core as for the algorithms presented in \cite{bauso2015distributed} and \cite{raja2020approachability}.
After receiving the required information, agent $i$ projects the average $\hat{\boldsymbol{x}}_i^{k}:= \sum_{j \in \mathcal{E}_{i,j}} w_{i,j}^k \boldsymbol{x}_{j}^{k} $ on the set of favorable payoffs $\mathcal{X}_i$. Thus, the iteration reads as
$\boldsymbol{x}_i^{k+1}= \mathrm{proj}_{\mathcal{X}_i}\textstyle (\hat{\boldsymbol{x}}_i^{k}). $
We can generalize this iteration by replacing the projection operator, $\mathrm{proj}(\cdot)$, with a special class of operators namely, paracontractions. This generalization enables us to utilize operator theory for showing the convergence of our algorithm later. The protocol we propose for an agent $i \in \mathcal{E}_{i,j}$ is $\boldsymbol{x}_i^{k+1}= T_i\textstyle (\hat{\boldsymbol{x}}_i^{k}), $
that in collective form, for negotiation between pair $(i,j)$, reads as the fixed-point iteration
 \begin{equation}\label{main_it}
     \boldsymbol{x}_{(i,j)}^{k+1} = \boldsymbol{T} (\boldsymbol{W}^k \boldsymbol{x}_{(i,j)}^{k}),
\end{equation}
 where  $\boldsymbol{T} (\boldsymbol{x}):= \mathrm{col}(T_1(\boldsymbol{x}_1),  T_2(\boldsymbol{x}_2))  $ and $\boldsymbol{W}^k := W^{k} \otimes I_4 $ represents an adjacency matrix. In (\ref{main_it}), we require the operator $T_i$ to have $\mathcal{X}_i$ in (\ref{eq: fav payoff}) as fixed-point set, i.e., $ \mathrm{fix}(T_i) = \mathcal{X}_i $.
 \begin{assum}[Paracontractions]\label{asm: fixed points of T}
 The operator $\boldsymbol{T}$ in (\ref{main_it}) is such that $T_i \in \mathcal{T}$, where $\mathcal{T}$ is a finite family of paracontraction operators such that $ \mathrm{fix}(T_i) = \mathcal{X}_i$ with $\mathcal{X}_i$ in (\ref{eq: fav payoff}). $\hfill \square$
 \end{assum}
The iteration in (\ref{main_it}) provides the bilateral negotiation process that will be executed by each matched pair $(i,j) \in \mathcal{I}_{B} \times \mathcal{I}_{S}$ independently. Next, we provide our main convergence result.


\begin{theorem}[Convergence of bilateral negotiation]\label{theorem: main}
Let Assumptions \ref{asm: graph} $-$ \ref{asm: fixed points of T} hold. Let $\mathcal{X}_{(i,j)}:=\mathcal{X}_i \cap \mathcal{X}_j $ with $\mathcal{X}_i$ as in (\ref{eq: fav payoff}). Then, starting from any $\boldsymbol{x}^0_{(i,j)}$, the sequence \((\boldsymbol{x}^{k}_{(i,j)})_{k=0}^{\infty}\) generated by the iteration in (\ref{main_it}) converges to $\boldsymbol{x}^*_{(i,j)} \in \mathcal{X}^2_{(i,j)} \cap \mathcal{A}$ with $\mathcal{A}$ as in (\ref{eq: consensus}) and  the collective payoff vector  $\boldsymbol{x}^* \in \prod_{{(i,j) \in \mathcal{E}}} \mathcal{X}_{(i,j)}$ is the $\tau$-value in (\ref{eq: tau value}) thus belongs to the core, $\mathcal{C}_M$ in (\ref{core assignment}). $\hfill \square$
\end{theorem}

\section{Numerical simulations}
In this section, we demonstrate the effectiveness of proposed algorithm by conducting numerical simulations of our bilateral P2P market design for time slots with high PV generation. We consider 3 residential prosumers with energy deficiency and  3  with  surplus  to  act  as  buyers  and  sellers,  respectively. Sellers are equipped with either a PV source of capacity 2 - 5 kWp or an ES of 4 kWh.  The buyers assign preference level to each  seller using the priority factor $\alpha_{i,j} \in [1,1.5]$ with $1$ being indifference to any criteria, e.g. green source (PV) vs brown source (ES), etc. Furthermore, buyers choose base valuation of the energy $p_i$ such that their bid is higher than the grid's buying price $g_\text{b} = 0.05 $ £/kWh and not more than the grid's selling price $g_\text{s} = 0.17 $ £/kWh and the sellers choose their valuation $c_j$ less than the grid's selling price $g_\text{s}$  \cite{han2018incentivizing}. To account for the uncertainty, we use three PV generation scenarios.
\begin{figure}[t]
\centering
\includegraphics[width=0.85\linewidth]{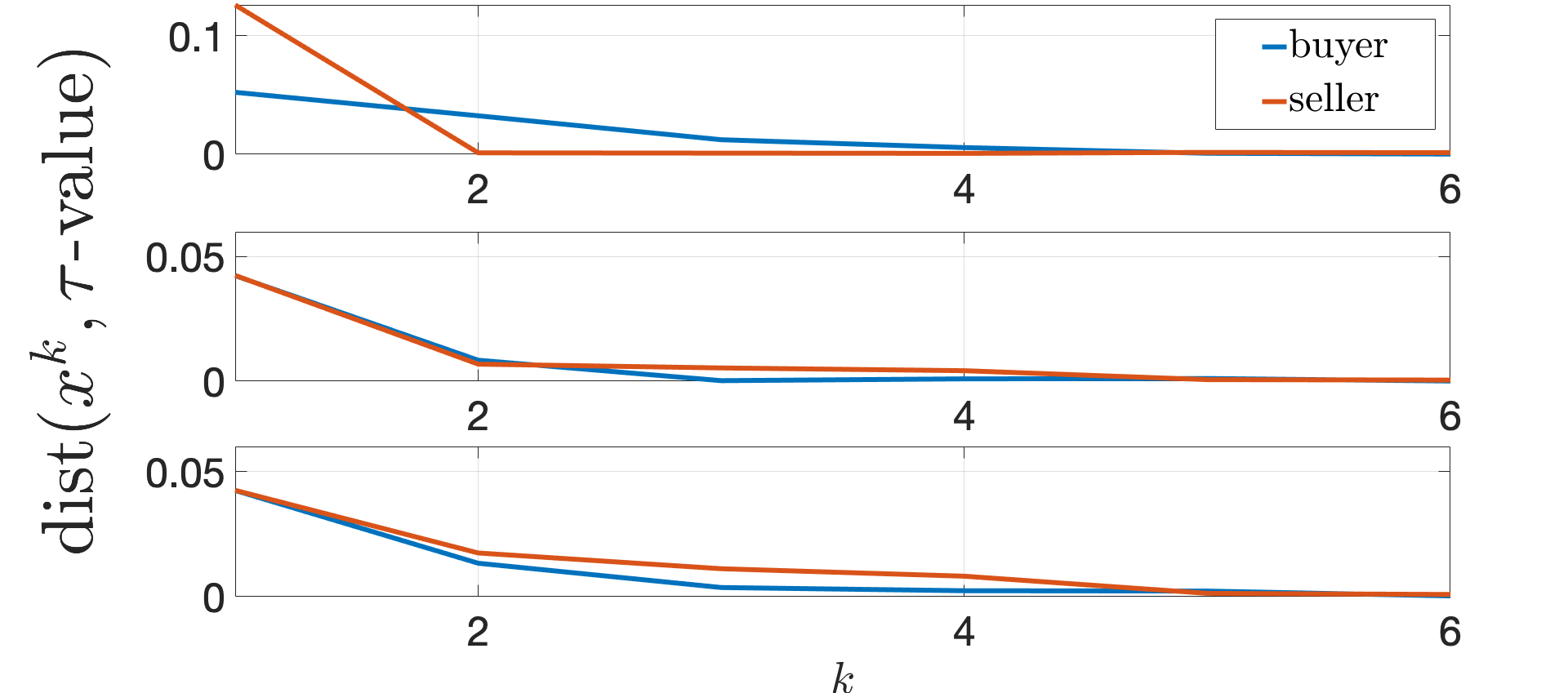}
\caption{Trajectories of $\mathrm{dist}(\boldsymbol{x}^k, \tau$-value$)$ via bilateral negotiation algorithm with operator $T_i := \mathrm{proj}_{\mathcal{X}_i}$ for optimally matched buyer-seller pairs.}
\label{fig: convergence}
\vspace{-3mm}
\end{figure}

\begin{figure}[t]
\centering
\includegraphics[width=0.8\linewidth]{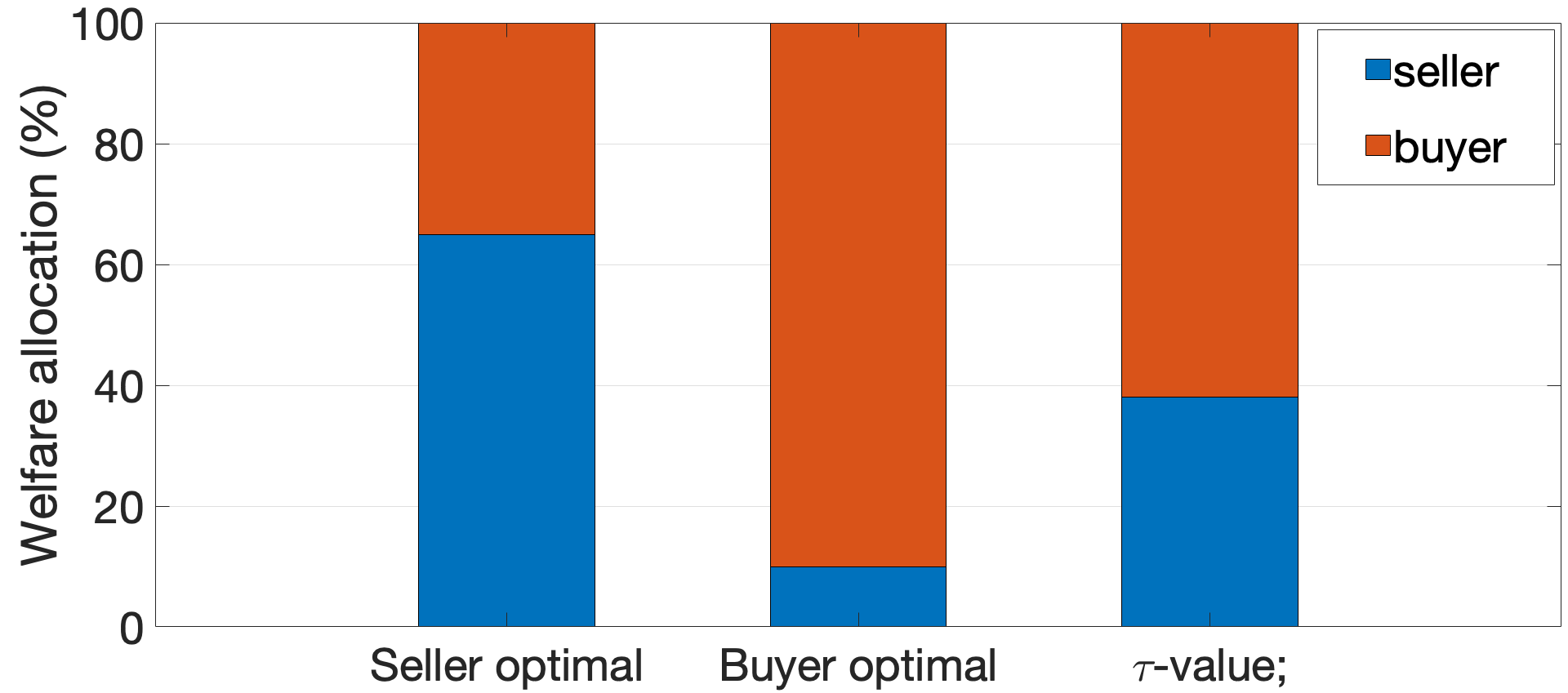}
\caption{Percentage distribution of total welfare between buyers and sellers via seller optimal, buyer optimal and $\tau$-value payoffs.}
\label{fig: welfare_assigned}
\vspace{-5mm}
\end{figure}
In our P2P market setup, first the market operator evaluates the optimal trading pairs of buyers and sellers using the formulation in (\ref{eq: assignment game}). Then, each matched pair internally negotiates for the contract prices via mechanism in (\ref{main_it}). In Figure \ref{fig: convergence}, we show the convergence of bilateral negotiation algorithm to the respective $\tau$-value payoffs. In Figure \ref{fig: welfare_assigned} we show the welfare allocation by respective points in the core. As the payoff is negotiated bilaterally the gain of buyer corresponds to the loss of seller and vice versa thus, we observe that, the $\tau$-value payoff provides more fair treatment to both sides. Next, we illustrate the economic benefit of trading inside the P2P market, compared to trading with the grid, in Figure \ref{fig: revenue}.
\begin{figure}[t]
\centering
\includegraphics[width=0.85\linewidth]{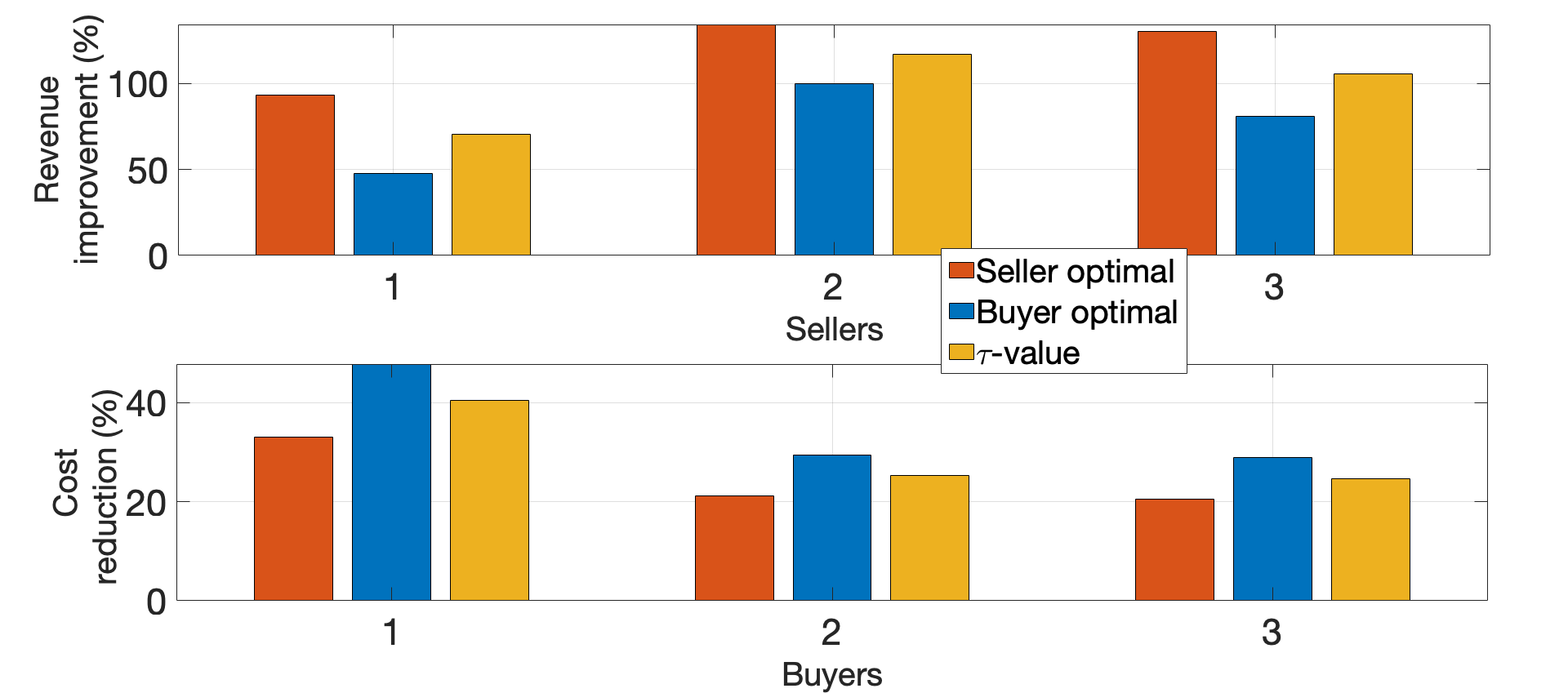}
\caption{Average revenue improvement (sellers) and average cost reduction (buyers) via seller optimal, buyer optimal and $\tau$-value payoffs compared to energy trading with the grid.}
\label{fig: revenue}
\vspace{-5mm}
\end{figure}
\section{Conclusion}
We have modelled P2P energy trading as an assignment game (coalitional game) and proposed a bilateral negotiation process as a clearing mechanism. The proposed P2P electricity market model encourages prosumers to participate by providing ease of accessibility, flexibility of choice and economic benefits, i.e., higher revenue (sellers) and lower energy costs (buyers) compared to trading with the grid. Furthermore, the bilateral negotiation mechanism enables participants to reach a trading contract ($\tau$-value) which fairly divides the resulting market welfare among buyers and sellers.\\
\vspace{-5mm}
\appendix
To prove the convergence of (\ref{main_it}), as stated in Theorem \ref{theorem: main}, we first provide useful property of a paracontraction operator.
\begin{lem}[\cite{Fullmer2018}, Thm. 2]\label{lemma: paracontractions}
Let $\mathcal{M} = \{M_1, \ldots, M_m\}$ be a set of paracontractions such that $\bigcap_{M \in \mathcal{M}} \mathrm{fix}(M) \neq \varnothing $. Let the communication graph be connected and consider the iteration $\boldsymbol{x}^{k+1} = \boldsymbol{M}(\boldsymbol{W}^k(\boldsymbol{x}^k)),$ where $\boldsymbol{M} (\boldsymbol{x}):= \mathrm{col}(M_1(\boldsymbol{x}_1), \ldots,  M_m(\boldsymbol{x}_m))$. Then, the state $\boldsymbol{x}^{k}$ converges to a state in the set $\mathcal{A} \cap \mathrm{fix}(\boldsymbol{M})$ as $k \to \infty$. $\hfill \square$
\end{lem}
\begin{proof}
By (\ref{eq: fav payoff}), $\mathcal{X}_{(i,j)} = \frac{(\overline{x}_i', \underline{x}_j'') + (\underline{x}_i', \overline{x}_j'')}{2} $. Let Assumptions \ref{asm: graph} and \ref{asm: fixed points of T} hold then, by Lemma \ref{lemma: paracontractions} the iteration in (\ref{main_it}) converges to $\mathcal{A} \cap \mathcal{X}^2_{(i,j)}$. Next, by Definition in (\ref{eq: tau value}) the collective payoff $\boldsymbol{x}^*$ is the $\tau$-value thus, by Remark \ref{rem: tau value in core}, $\boldsymbol{x}^* \in \mathcal{C}_M$.
\end{proof}

\bibliographystyle{IEEEtran}
\bibliography{IEEEabrv,Bibliography}

\begin{thebibliography}{10}
\providecommand{\url}[1]{#1}
\csname url@samestyle\endcsname
\providecommand{\newblock}{\relax}
\providecommand{\bibinfo}[2]{#2}
\providecommand{\BIBentrySTDinterwordspacing}{\spaceskip=0pt\relax}
\providecommand{\BIBentryALTinterwordstretchfactor}{4}
\providecommand{\BIBentryALTinterwordspacing}{\spaceskip=\fontdimen2\font plus
\BIBentryALTinterwordstretchfactor\fontdimen3\font minus
  \fontdimen4\font\relax}
\providecommand{\BIBforeignlanguage}[2]{{%
\expandafter\ifx\csname l@#1\endcsname\relax
\typeout{** WARNING: IEEEtran.bst: No hyphenation pattern has been}%
\typeout{** loaded for the language `#1'. Using the pattern for}%
\typeout{** the default language instead.}%
\else
\language=\csname l@#1\endcsname
\fi
#2}}
\providecommand{\BIBdecl}{\relax}
\BIBdecl

\bibitem{cai2013impact}
D.~W. Cai, S.~Adlakha, S.~H. Low, P.~De~Martini, and K.~M. Chandy, ``Impact of
  residential pv adoption on retail electricity rates,'' \emph{Energy Policy},
  vol.~62, pp. 830--843, 2013.

\bibitem{Tushar2020}
W.~Tushar, T.~K. Saha, C.~Yuen, D.~Smith, and H.~V. Poor, ``{Peer-to-Peer
  Trading in Electricity Networks: An Overview},'' \emph{IEEE Transactions on
  Smart Grid}, vol.~11, no.~4, pp. 3185--3200, 2020.

\bibitem{moret2018energy}
F.~Moret and P.~Pinson, ``Energy collectives: a community and fairness based
  approach to future electricity markets,'' \emph{IEEE Transactions on Power
  Systems}, vol.~34, no.~5, pp. 3994--4004, 2018.

\bibitem{mengelkamp2018designing}
E.~Mengelkamp, J.~G{\"a}rttner, K.~Rock, S.~Kessler, L.~Orsini, and
  C.~Weinhardt, ``Designing microgrid energy markets: A case study: The
  brooklyn microgrid,'' \emph{Applied Energy}, vol. 210, pp. 870--880, 2018.

\bibitem{tushar2018peer}
W.~Tushar, T.~K. Saha, C.~Yuen, P.~Liddell, R.~Bean, and H.~V. Poor,
  ``Peer-to-peer energy trading with sustainable user participation: A game
  theoretic approach,'' \emph{IEEE Access}, vol.~6, pp. 62\,932--62\,943, 2018.

\bibitem{han2018incentivizing}
L.~Han, T.~Morstyn, and M.~McCulloch, ``Incentivizing prosumer coalitions with
  energy management using cooperative game theory,'' \emph{IEEE Transactions on
  Power Systems}, vol.~34, no.~1, pp. 303--313, 2018.

\bibitem{tushar2019grid}
W.~Tushar, T.~K. Saha, C.~Yuen, T.~Morstyn, H.~V. Poor, R.~Bean \emph{et~al.},
  ``Grid influenced peer-to-peer energy trading,'' \emph{IEEE Transactions on
  Smart Grid}, vol.~11, no.~2, pp. 1407--1418, 2019.

\bibitem{tushar2020coalition}
W.~Tushar, T.~K. Saha, C.~Yuen, M.~I. Azim, T.~Morstyn, H.~V. Poor, D.~Niyato,
  and R.~Bean, ``A coalition formation game framework for peer-to-peer energy
  trading,'' \emph{Applied Energy}, vol. 261, p. 114436, 2020.

\bibitem{khorasany2020new}
M.~Khorasany, A.~Paudel, R.~Razzaghi, and P.~Siano, ``A new method for peer
  matching and negotiation of prosumers in peer-to-peer energy markets,''
  \emph{IEEE Transactions on Smart Grid}, vol.~12, no.~3, pp. 2472--2483, 2020.

\bibitem{shapley1971assignment}
L.~S. Shapley and M.~Shubik, ``The assignment game i: The core,''
  \emph{International Journal of game theory}, vol.~1, no.~1, pp. 111--130,
  1971.

\bibitem{Nunez2002}
M.~N{\'{u}}{\~{n}}ez and C.~Rafels, ``{The assignment game: The
  $\tau$-value},'' \emph{International Journal of Game Theory}, vol.~31, no.~3,
  pp. 411--422, 2002.

\bibitem{bauso2015distributed}
D.~Bauso and G.~Notarstefano, ``Distributed $ n $-player approachability and
  consensus in coalitional games,'' \emph{IEEE Transactions on Automatic
  Control}, vol.~60, no.~11, pp. 3107--3112, 2015.

\bibitem{raja2020approachability}
A.~A. Raja and S.~Grammatico, ``On the approachability principle for
  distributed payoff allocation in coalitional games,''
  \emph{IFAC-PapersOnLine}, vol.~53, no.~2, pp. 2690--2695, 2020.

\bibitem{Fullmer2018}
D.~Fullmer and A.~S. Morse, ``{A Distributed Algorithm for Computing a Common
  Fixed Point of a Finite Family of Paracontractions},'' \emph{IEEE
  Transactions on Automatic Control}, vol.~63, no.~9, pp. 2833--2843, 2018.

\end{thebibliography}

\end{document}